\documentclass[12pt]{article}
\usepackage{graphicx,amsfonts,amsmath,amssymb,amsthm,mathrsfs,mathabx,url,hyperref}
\usepackage[usenames,dvipsnames]{xcolor}
\usepackage[toc,page]{appendix}

\title{Bogoliubov Transformations Beyond Shale--Stinespring: Generic $ v^* v $ for bosons}
\author{
Sascha Lill\footnote{Universit\`a degli Studi di Milano, Dipartimento di Matematica, Via Cesare Saldini 50, 20133 Milano, Italy}\ \footnote{E-Mail: \url{sascha.lill@unimi.it}}~
} 
\date{\today}

\usepackage[utf8]{inputenc}
\usepackage{tikz}
\usetikzlibrary{decorations.pathreplacing}

\usepackage[section]{placeins}
\usepackage{capt-of, caption}
\usepackage{paralist}
\usepackage{booktabs}
\usepackage{hhline}
\usepackage{color}                  

\newcommand{\be}{\boldsymbol{e}}
\newcommand{\bof}{\boldsymbol{f}}
\newcommand{\bg}{\boldsymbol{g}}

\newcommand{\bphi}{\boldsymbol{\phi}}

\newcommand{\bpsi}{\boldsymbol{\psi}}

\newcommand{\cA}{\mathcal{A}}
\newcommand{\cB}{\mathcal{B}}
\newcommand{\cD}{\mathcal{D}}
\newcommand{\cE}{\mathcal{E}}
\newcommand{\cI}{\mathcal{I}}

\newcommand{\cO}{\mathcal{O}}

\newcommand{\cQ}{\mathcal{Q}}

\newcommand{\cV}{\mathcal{V}}
\newcommand{\fc}{\mathfrak{c}}

\newcommand{\fr}{\mathfrak{r}}

\newcommand{\fR}{\mathfrak{R}}
\newcommand{\sA}{\mathscr{A}}

\newcommand{\sF}{\mathscr{F}}
\newcommand{\sH}{\mathscr{H}}

\newcommand{\CCC}{\mathbb{C}}

\newcommand{\NNN}{\mathbb{N}}

\newcommand{\UUU}{\mathbb{U}}

\newcommand{\dom}{{\rm dom}}
\newcommand{\eRen}{\mathrm{eRen}}
\newcommand{\ex}{{\rm ex}}

\newcommand{\Ren}{\mathrm{Ren}}

\newcommand{\tr}{\mathrm{tr}}

\newcommand{\FB}{\overline{\mathscr{F}}}
\newcommand{\EB}{\overline{\mathcal{E}_\mathscr{F}}}

\newcommand{\EBn}{\overline{\mathcal{E}_{\mathscr{F}, 0}}}
\newcommand{\cAB}{\overline{\mathcal{A}}}

\newtheorem{lemma}{Lemma}
\numberwithin{lemma}{section}

\numberwithin{corollary}{section}
\newtheorem{theorem}{Theorem}
\numberwithin{theorem}{section}

\numberwithin{proposition}{section}

\numberwithin{remark}{section}
\theoremstyle{definition}\newtheorem{definition}{Definition}
\numberwithin{definition}{section}
\theoremstyle{definition}
\numberwithin{assumption}{section}

\newcounter{remarks}
\begin{document}
\maketitle
\begin{abstract}
We construct an extension of Fock space and prove that it allows for implementing bosonic Bogoliubov transformations in a certain extended sense. While an implementation in the regular sense on Fock space is only possible if a certain operator $ v^* v $ is trace class (this is the well--known Shale--Stinespring condition), the extended implementation works without any restrictions on this operator. This generalizes a recent result of extended implementability, which required $ v^* v $ to have discrete spectrum.\\
\end{abstract}

\tableofcontents

\section{Introduction}	
\label{sec:intro}

Bogoliubov transformations are a versatile tool for elucidating the physical properties of many--body and quantum field theory (QFT) models in mathematical physics. Applications include the simplification of Hamiltonians in interacting bosonic or fermionic $ N $--body systems \cite{Haag1962, BachLiebSolovej1994, BenedikterNamPortaSchleinSeiringer2020, FalconiGiacomelliHainzlPorta2021, AdhikariBrenneckeSchlein2021}, QFT toy models \cite{Berezin1992, Marecki2003, Derezinski2017} and relativistic quantum dynamics \cite{TorreVaradarajan1999, HollandsWald2001}.\\
In simple words, a Bogoliubov transformation $ \cV $ is a linear replacement of creation-- and annihilation operators $ a^\dagger \mapsto b^\dagger, a \mapsto b $ such that $ b, b^\dagger $ satisfy the same commutation relations as $ a, a^\dagger $. There are cases where $ \cV $ allows for finding a unitary operator on Fock space $ \UUU_\cV: \sF \to \sF $ such that $ \UUU_\cV a^\dagger \UUU_\cV^{-1} = b^\dagger $ and $ \UUU_\cV a \UUU_\cV^{-1} = b $. In this case, $ \UUU_\cV $ is called an \textbf{implementer} of $ \cV $ and $ \cV $ is called \textbf{implementable}. The case of implementability is interesting for the following reason: Suppose, we are given a Hamiltonian $ H $, which is an inconvenient sum of products of $ a, a^\dagger $, but there exists a Bogoliubov transformation $ \cV $ and a constant $ c \in \CCC $, such that $ (H + c) $ takes a much more convenient form in the $ b, b^\dagger $--operators. This situation arises frequently for quadratic Hamiltonians, see Remark \ref{rem:diagonalization}. Then, $ \cV^{-1} $, which is implemented by $ \UUU_\cV^{-1} = \UUU_\cV^* $, replaces $ b, b^\dagger $ by $ a, a^\dagger $. So, under the replacement, $ (H + c) $ becomes
\begin{equation}
	\widetilde{H} := \UUU_\cV^{-1} (H + c) \UUU_\cV,
\label{eq:Htrafo}
\end{equation}
which has a convenient form in the $ a, a^\dagger $--operators. This convenient form allows for a simple description of the dynamics generated by $ \widetilde{H} $, which is by \eqref{eq:Htrafo} unitarily equivalent to the dynamics generated by $ H $.\\

The important question when $ \cV $ is implementable has long been settled by Shale and Stinespring \cite{Shale1962, ShaleStinespring1965}: This is the case if and only if for some operator $ v $ \eqref{eq:bdaggerb}, characterizing $ \cV $, we have 
\begin{equation*}
	\tr(v^* v) < \infty.
\end{equation*}
This so--called Shale--Stinespring condition is rather restrictive. Therefore, the author recently \cite{BBS} proposed the construction of an \textbf{implementer in an extended sense}, $ \UUU_\cV $, which also achieves the replacement $ a^\sharp \mapsto b^\sharp $ induced by $ \cV $, but can be constructed for a much greater class of Bogoliubov transformations. In \cite{BBS}, two Fock space extensions $ \widehat{\sH}, \FB $ were constructed, each together with an operator $ \UUU_\cV: \cD_\sF \to \widehat{\sH} $ or $ \UUU_\cV: \cD_\sF \to \FB $, with $ \cD_\sF \subset \sF $ being a dense subspace, such that \eqref{eq:Htrafo} holds as a strong operator identity on $ \cD_\sF $. Here, $ \widehat{\sH} $ is an infinite tensor product space of the form introduced by von Neumann \cite{vonNeumann1939} and $ \FB $ is an ``extended state space'' as recently introduced in \cite{ESS}. The main result in \cite{BBS} was that $ \cV $ can be implemented in the extended sense, whenever
\begin{equation*}
	v^* v \text{ has \textbf{discrete spectrum}}.
\end{equation*}
In this general case, $ c $ becomes an ``infinite renormalization constant'', which was rigorously interpreted in \cite{BBS} as an element of some vector space $ \Ren_1 $.\\
The purpose of this article is to propose a new Fock space extension $ \EB $, related to $ \FB $, which allows for an extended implementer $ \UUU_\cV: \cD_\sF \to \EB $
\begin{equation*}
	\text{in the \textbf{bosonic} case for \textbf{arbitrary} } v^* v,
\end{equation*}
as long as $ v $ is defined on a suitable domain. We establish the notation in Section \ref{sec:notation}, construct $ \EB $ in Section \ref{sec:ESS}, state our main result, Theorem \ref{thm:bosonicESS}, in Section \ref{sec:mainresult} and prove it in Section \ref{sec:proofmainresult}.\\
An analogous result can be expected to hold for fermions, as long as 1 is not an eigenvalue of infinite multiplicity of $ v^* v $, i.e., $ \cV $ only performs a particle--hole transformation on finitely many modes. The main technical complication in the fermionic case comes from the fact that the operator $ \cO $ \eqref{eq:cOpolar} is unbounded, whereas it is bounded in the bosonic case.\\
An extended implementer $ \UUU_\cV $ is particularly useful in cases where $ H $ is an algebraic expression that does not define an operator on a dense subspace of $ \sF $, while $ \widetilde{H} = \UUU_\cV^{-1} (H + c) \UUU_\cV $ maps $ \cD_\sF $ into itself and allows for a self--adjoint extension. In that case, we can mathematically make sense of $ (H + c) $ as an operator mapping $ \UUU_\cV(\cD_\sF) $ into itself (see Figure \ref{fig:Dressing_Fockspace_EB}) and use the dynamics generated by $ \widetilde{H} $ on $ \sF $ for a physical interpretation.\\
The construction of $ \EB $ is significantly shorter than that of $ \FB $ in \cite{BBS}. We intend a future use of $ \EB $ as a bookkeeping tool for more general operator transformations that go beyond Bogoliubov transformations.\\

\begin{figure}
	\centering
	\scalebox{0.95}{\begin{tikzpicture}

\node[anchor = east] at (0,0)  {$ \sF \supset \cD_\sF $};
\node[anchor = east] at (0.5,1.5) {$ \EB \supset \UUU_\cV (\cD_\sF) $};
\draw[thick, ->] (-0.3,0.3) --node[anchor = west] {$\UUU_\cV$} ++(0,0.9);
\draw[thick, ->] (0.2,0) -- node[anchor = south] {$\widetilde{H}$}  (1.8,0);

\node[anchor = west] at (2,0) {$ \cD_\sF \subset \sF $};
\node[anchor = west] at (1.5,1.5) {$ \UUU_\cV(\cD_\sF) \subset \EB $};
\draw[thick,<-] (2.3,0.3) --node[anchor = west] {$\UUU_\cV^{-1}$} ++(0,0.9);
\draw[thick, ->] (0.4,1.5) -- node[anchor = south] {$ (H + c) $}  (1.6,1.5);

\end{tikzpicture}}
	\scalebox{0.95}{\begin{tikzpicture}
 \filldraw[thick,fill = blue!10!white] (4,-0.2) rectangle ++ (2,1.4);
 \filldraw[thick,fill = gray!30!white] (4.9,0.5) ellipse (0.8 and 0.4) node {$ \mathcal{D}_\mathscr{F} $};
 \draw[blue] (4.5,1.2) -- ++(0.5,0.3) node[anchor = south] {Fock space $ \mathscr{F} $};
 
 \filldraw[fill = blue!05!white] (6.7,-0.3) rectangle ++(4.5,2.3);
 \filldraw[thick,fill = blue!10!white] (7,-0.2) rectangle ++ (2,1.4);
 \filldraw[thick,dashed,fill = gray!30!white] (10.2,1.2) ellipse (0.8 and 0.4) node {$ \mathbb{U}_{\mathcal{V}}(\mathcal{D}_\mathscr{F}) $};
 \filldraw[thick,fill = gray!30!white] (7.9,0.5) ellipse (0.8 and 0.4) node {$ \mathcal{D}_\mathscr{F} $};
 \draw[line width = 2,->] (9.5,1.4) .. controls (9.2,1.5) and (8.8,1.4) .. (8.5,0.8);
 \node at (8.8,1.6) {$ \mathbb{U}_{\mathcal{V}}^{-1} $};
 \draw[line width = 2,<-] (9.5,1.2) .. controls (9.2,1.3) and (8.8,1.1) .. (8.6,0.6);
 \node at (9.3,0.8) {$ \mathbb{U}_{\mathcal{V}} $};
 \draw[line width = 2,->] (10.7,0.9) .. controls (11.8,0.5) and (11.8,1.6) .. (10.8,1.4);
 \node at (11.8,1.6) {$ (H + c) $};
 \draw[blue] (11.2,0) -- ++(0.4,0.3) node[anchor = west] {$ \overline{\mathcal{E}_\mathscr{F}} $};

\end{tikzpicture}}
	\caption{Left: $ \widetilde{H} $ is defined in \eqref{eq:Htrafo} such that the diagram commutes.\\ Right: Sketch of an extended implementation using the Fock space extension $ \EB $.}
	\label{fig:Dressing_Fockspace_EB}
\end{figure}

\section{Notation}	
\label{sec:notation}

The notation is adapted from \cite{BBS} in great parts. We consider a system with an indeterminate number of particles $ N $, whose one--particle sector is given by the sequence space $ \ell^2 $.\footnote{Note that any separable Hilbert space can be identified with $ \ell^2 $ by fixing an orthonormal basis, so the description with $ \ell^2 $ is quite general.} The configuration of the system is given by a vector of mode numbers $ (j_1, \ldots, j_N) $, which is an element of \textbf{configuration space}\footnote{The symbol $ \sqcup $ denotes the disjoint union of two sets. This notation is chosen to emphasize that, e.g., $ \NNN \sqcup \NNN^2 \sqcup \ldots \sqcup \NNN^N $ is not just $ \NNN^N $.}
\begin{equation}
	\cQ := \bigsqcup_{N \in \NNN_0} \cQ^{(N)} := \bigsqcup_{N \in \NNN_0} \NNN^N.
\label{eq:cQX}
\end{equation}
Here, $ \cQ^{(N)} = \NNN^N $ is also called the $ N $--particle sector of configuration space. The state of the system is described by a normed vector in the (mode--) \textbf{Fock space}
\begin{equation}
	\sF := L^2(\cQ).
\end{equation}
We impose bosonic symmetry using the symmetrization operator $ S_+: \sF \to \sF $ with
\begin{equation}
	(S_+ \Psi)(j_1, \ldots, j_N) := \frac{1}{N!} \sum_{\sigma \in S_N} \Psi(j_{\sigma(1)}, \ldots, j_{\sigma(N)}),
\label{eq:symmetricdef}
\end{equation}
where $ S_N $ is the permutation group. The bosonic Fock space is then given by
\begin{equation}
	\sF_+ := S_+ (\sF).
\label{eq:symmetricfockdef}
\end{equation}
Equivalently, using the symmetric tensor product
\begin{equation}
	\phi \otimes_S \phi := S_+ (\phi \otimes \phi),
\label{eq:otimesSA}
\end{equation}
we can write
\begin{equation}
	\sF_+ = \bigoplus_{N \in \NNN_0} (\ell^2)^{\otimes_S N}.
\label{sFtensorproduct}
\end{equation}
Denote by $ (\be_j)_{j \in \NNN} $ the canonical basis of $ \ell^2 $ and by $ a^\dagger_j, a_j $ the creation and annihilation operators corresponding to $ \be_j $, which are a priori just symbolic expressions. Products of these expressions or countable sums thereof are not necessarily defined as Fock space operators. Instead, we interpret them as elements of the $ ^* $--algebra $ \cAB $, defined as follows: Denote by
\begin{equation}
	\Pi_{\be} := \big\{ P_{\be} = a_{j_1}^{\sharp_1} \ldots a_{j_m}^{\sharp_m} \; \big\vert \; m \in \NNN_0, j_{\ell} \in \NNN, \sharp_\ell \in \{\cdot, \dagger\} \big\}
\end{equation}
the set of all finite operator products with respect to the basis $(\be_j)_{j \in \NNN} $. Then, $ \cAB $ is defined as the set of all countable sums
\begin{equation}
	\cAB := \left\{ H = \sum_{P_{\be} \in \Pi_{\be}} H_{j_1, \sharp_1, \ldots, j_m, \sharp_m} P_{\be} \; \middle\vert \; H_{j_1, \sharp_1, \ldots, j_m, \sharp_m} \in \CCC \right\},
\label{eq:cAB}
\end{equation}
with $ m $ depending on $ P_{\be} $. Let us denote the space of all complex--valued sequences by
\begin{equation}
	\cE = \{ \NNN \to \CCC \}.
\end{equation}
For $ \bphi \in \cE $, we then introduce the creation/annihilation operators
\begin{equation}
	a^\dagger(\bphi) := \sum_{j \in \NNN} \phi_j a^\dagger_j, \qquad
	a(\bphi) := \sum_{j \in \NNN} \overline{\phi_j} a_j, \qquad
	a^\dagger(\bphi), a(\bphi) \in \cAB,
\end{equation}
where the overline denotes complex conjugation. Alternatively, complex conjugation is described by the complex conjugation operator
\begin{equation}
	J: \cE \to \cE, \qquad (J \bphi)_j = \overline{\phi_j},
\label{eq:J}
\end{equation}
with $ J^2 = 1 $ and whose adjoint satisfies $ (J^*)^2 = 1 $. We will interchangeably use the notations $ \phi_j, (\bphi)_j $ and $ (\bphi)(j) $ for elements of a series $ \bphi \in \cE $. Further, for $ \bphi, \bpsi \in \ell^2 $ we impose the canonical commutation relations (CCR) on $ \cAB $
\begin{equation}
	[a(\bphi), a(\bpsi)] = [a^\dagger(\bphi), a^\dagger(\bpsi)] = 0, \qquad
	[a(\bphi), a^\dagger(\bpsi)] = \langle \bphi, \bpsi \rangle.
\label{eq:CCR}
\end{equation}
In case $ \bphi, \bpsi \in \ell^2 $, we may also densely define $ a^\sharp(\bphi) $ as Fock space operators
\begin{equation}
\begin{aligned}
	(a^\dagger(\bphi) \Psi)(j_1, \ldots, j_N) &= \frac{1}{\sqrt{N}} \sum_{\ell = 1}^N \phi_{j_\ell} \Psi(j_1, \ldots, j_{\ell - 1}, j_{\ell + 1}, \ldots, j_N)\\
	(a(\bphi) \Psi)(j_1, \ldots, j_N) &= \sqrt{(N+1)} \sum_{j \in \NNN} \overline{\phi_j} \Psi(j_1, \ldots, j_N, j),\\
\end{aligned}
\label{eq:aadaggermomentumspace}
\end{equation}
which preserve symmetry and satisfy \eqref{eq:CCR} as a strong operator identity on a dense domain in Fock space.\\
Creation and annihilation operators $ a^\sharp(\bphi) $ are particularly easy to handle, if the sum over $ j $ is finite, or equivalently, if the form factor is an element of the following space:
\begin{equation}
	\cD := \{ \bphi \in \ell^2 \; \mid \; \phi_j = 0 \; \text{for all but finitely many} \; j \in \NNN\}.
\label{eq:cD}
\end{equation}
Note that $ \cD \subset \ell^2 \subset \cE $ and $ \cE = \cD' $, i.e., $ \cE $ is the dual space of $ \cD $ with respect to a suitable seminorm--induced topology, see \cite[Chap.~1, Example~1.6]{BerezanskyKondratiev}.\\
By an \textbf{algebraic Bogoliubov} transformation, we mean a map $ \cV_\cA: \cAB \to \cAB $, that sends
\begin{equation}
\begin{aligned}
	a^\dagger(\bphi) &\mapsto b^\dagger(\bphi) := a^\dagger(u \bphi) + a(v \overline{\bphi}) , \qquad\\
	a(\bphi) &\mapsto b(\bphi) := a^\dagger(v \overline{\bphi}) + a(u \bphi)
\end{aligned}
\label{eq:bdaggerb}
\end{equation}
for all $ \bphi \in \cD $, where $ u, v: \cD \to \ell^2 $ satisfy the \textbf{bosonic Bogoliubov relations} as a weak operator identity:
\begin{equation}
\begin{aligned}
	u^*u - v^T \overline{v} &= 1 \qquad
	&u^*v - v^T \overline{u} = 0\\
	u u^* - v v^* &= 1 \qquad
	&u v^T - v u^T =0.\\
\end{aligned}
\label{eq:Bogoliuborelations}
\end{equation}
Here, $ u^* $ is the adjoint of $ u $ and the conjugate and transposed operators are given by $ \overline{u} = J u J $ and $ u^T = J u^* J $ (and the same for $ v $). Writing $ u, v $ as matrices of infinite size, $ u = (u_{jk})_{j,k \in \NNN}, v = (v_{jk})_{j,k \in \NNN} $, we equivalently have $ (u^T)_{jk} = (u)_{kj} $, $ (u^*)_{jk} = \overline{(u)_{kj}} $ and $ (\overline{u})_{jk} = \overline{(u)_{jk}} $. Further, we may keep track of $ \cV_\cA $ by a block matrix $ \cV = \big( \begin{smallmatrix} u & v \\ \overline{v} & \overline{u} \end{smallmatrix} \big) $. In \cite[Lemma~4.1]{BBS} it was proven that \eqref{eq:Bogoliuborelations} holds, whenever $ \cV $ and $ \cV^* $ both preserve the CCR, that is, \eqref{eq:CCR} also holds for $ a^\sharp $ replaced by $ b^\sharp $.\\

\section{Construction of the Fock space extension}
\label{sec:ESS}

We will establish implementability on a Fock space extension $ \EB \supset \sF $ here, which is related, but not equal to the extensions $ \FB, \FB_{\ex} $ in \cite{BBS} and \cite{ESS}. Just as in \cite{BBS}, we define the \textbf{generalized $ N $--particle space} and the \textbf{generalized Fock space} as
\begin{equation}
	\cE^{(N)} := \big\{ \Psi^{(N)}: \NNN^N \to \CCC \big\}, \qquad
	\cE_\sF := \big\{ \Psi: \cQ \to \CCC \big\} \supset \sF.
\label{eq:cEsF}
\end{equation}
The definition of $ \EB $ is motivated by formal and possibly divergent sums that appear when formally applying one or several operators $ a(\bphi), \bphi \in \cE $ as in \eqref{eq:aadaggermomentumspace} to functions in $ \cE_\sF $. In its most general form, such a formal sum on the $ (N) $--sector reads
\begin{equation}
	\Psi^{(N)}(j_1, \ldots, j_N) = \sum_{j_{N + 1}, \ldots, j_{N + L} \in \NNN} \Psi_{(L)}^{(N + L)}(j_1, \ldots, j_{N + L}),
\label{eq:Psisum}
\end{equation}
where $ N, L \in \NNN_0 $. There are two ways to read \eqref{eq:Psisum}: First, $ \Psi^{(N)}(j_1, \ldots, j_N) $ can be viewed as a generalized function value, which is given by a formal sum at fixed $ (j_1, \ldots, j_N) $. Second, \eqref{eq:Psisum} can be seen as a definition for a generalized function on $ \NNN^N $ that is obtained by taking the ``raw functions'' $ \Psi_{(L)}^{(N + L)} $ and formally ``integrating out'' the last $ L $ variables.\\

Mathematically, the sector $ \Psi^{(N)} $ is specified by a pair $ \big( \Psi_{(L)}^{(N + L)}, L \big) $. In order to allow for countable linear combinations and products of such expressions, we introduce the algebra of functions
\begin{equation}
	\EBn := \Big\{ \NNN_0 \times \NNN_0 \to \cE_\sF \; \Big\vert \; (N, L) \mapsto \Psi_{(L)}^{(N + L)} \in \cE^{(N + L)} \Big\}.
\label{eq:EBn}
\end{equation}
Here, $ \EBn $ is to be understood as a $ \CCC $--vector space with linear combinations defined argument--wise, i.e.,
\begin{equation}
	(c\Psi + \Psi')_{(L)}^{(N + L)} = c (\Psi)_{(L)}^{(N + L)} + (\Psi')_{(L)}^{(N + L)}, \quad c \in \CCC.
\label{eq:EBnsum}
\end{equation}
The algebra multiplication is given by
\begin{equation}
	(\Psi \otimes \Psi')_{(L)}^{(N + L)} = \sum_{\substack{N_1 + N_2 = N \\ L_1 + L_2 = L}} (\Psi)_{(L_1)}^{(N_1 + L_1)} \otimes (\Psi')_{(L_2)}^{(N_2 + L_2)},
\label{eq:EBnmultiplication}
\end{equation}
with $ \otimes $ denoting the topological tensor product in $ \cE_\sF $. We now introduce the sum notation, where we write each function in $ \EBn $ as a formal countable sum
\begin{equation}
	\Psi^{(N)}(j_1, \ldots, j_N) = \sum_{L \in \NNN_0} \; \sum_{j_{N + 1}, \ldots, j_{N + L} \in \NNN} \Psi_{(L)}^{(N + L)}(j_1, \ldots, j_{N + L}).
\label{eq:sumnotation}
\end{equation}
This notation is consistent with the notation in \eqref{eq:Psisum}, and \eqref{eq:EBnmultiplication} coincides with the formal multiplication of two sums of the kind \eqref{eq:sumnotation}. However, permuting sum indices or executing convergent sums in \eqref{eq:sumnotation} intuitively leaves the sum invariant, but results in a different associated element in $ \EBn $. Therefore, we mod out a (two--sided) ideal\footnote{By an ideal of an algebra, we mean that $ \cI $ shall be closed (in an algebraic sense) under the (commutative) algebra multiplication $ \otimes $ and linearity. So for $ \Psi_1, \Psi_2 \in \cI, \Psi \in \EBn $ and $ c \in \CCC $, we have $ (\Psi \otimes \Psi_1) \in \cI $ and $ c \Psi_1 + \Psi_2 \in \cI $. Intuitively, the elements in $ \cI $ are the ones equivalent to 0.} $ \cI \subset \EBn $, which is generated by requiring that:
\begin{enumerate}[(A)]
\item $ \Psi - \Psi' \in \cI $, whenever each $ (\Psi')_{(L)}^{(N + L)} $ can be obtained by permuting the last $ L $ indices of $ (\Psi)_{(L)}^{(N + L)} $. So we may swap indices that are integrated out.\\

\item $ \Psi - \Psi' \in \cI $, whenever we may choose for each pair $ (N, L) $ some integer $ 0 \le \Delta L_{N, L} \le L $, such that
\begin{equation}
	(\Psi')_{(L')}^{(N + L')}(j_1, \ldots, j_{N + L'}) = \sum_{L: L' + \Delta L_{N, L} = L} \; \sum_{j_{N + L' + 1}, \ldots, j_{N + L}} (\Psi)_{(L)}^{(N + L)}(j_1, \ldots, j_{N + L})
\label{eq:PsiPsiprime2}
\end{equation}
is an absolutely convergent sum for all $ N, L' \in \NNN_0 $ and all $ (j_1, \ldots, j_{N + L'}) \in \NNN^{L'} $. In other words, we may execute absolutely convergent sums, where $ \Delta L_{N, L} $ is the number of sums executed in $ (\Psi)_{(L)}^{(N + L)} $.\\

\end{enumerate}
The \textbf{extended state space} is then defined as the quotient algebra
\begin{equation}
	\EB := \EBn /_{\cI}.
\label{eq:EB}
\end{equation}
Its elements (which are cosets of $ \cI $) can thus be treated as if they were formal sums \eqref{eq:sumnotation} and we adopt the sum notation also for elements of $ \EB $.\\

\paragraph{Remarks.} 

\begin{enumerate}
\setcounter{enumi}{\theremarks}
\item \label{rem:Fockspaceextension} \textit{$ \EB $ extends $ \sF $}: In sum notation \eqref{eq:sumnotation}, each $ \Psi \in \cE_\sF \supset \sF $ corresponds to exactly one element of $ \EBn $, viz. the one sending $ (N, 0) \mapsto \Psi^{(N)} $ and $ (N, L) \mapsto 0 $ for $ L \ge 1 $.\\
Further, if $ \Psi \in \cE_\sF $ with $ \Psi \neq 0 $, then the associated element in $ \EBn $ is not in $ \cI $. This is because \eqref{eq:sumnotation} renders a complex number at each fixed $ (j_1, \ldots, j_N) $, which is nonzero for at least one $ (j_1, \ldots, j_N) $ and cannot be changed by executing/un--executing convergent sums or permuting sum indices. So each $ \Psi \in \cE_\sF $ can be identified with a distinct coset of $ \cI $, i.e., with an element of $ \EB $. So we may embed $ \cE_\sF \supset \sF $ into $ \EB $ and thus, $ \EB $ is a true Fock space extension.\\

\item \label{rem:tensorproduct} \textit{Tensor products on $ \EB $}: The $ \otimes $--product on $ \EBn $ \eqref{eq:EBnmultiplication} induces an $ \otimes $--product on the quotient algebra $ \EB $. With the tensor product $ \otimes $, one can also view $ \EBn $ as a ``double--graded algebra'' with degrees $ N, L \in \NNN_0 $, just as $ \cE_\sF $ together with the tensor product $ \otimes $ becomes a graded algebra (with degree $ N $). Also, $ \EB $ is a graded algebra with degree $ N $, but not ``double--graded'', since $ L $ is not unique.\\

\item \label{rem:BBScomparison_Ren1} \textit{Comparison with \cite{BBS}}: The construction of $ \EB $ deviates from the construction of the Fock space extensions $ \FB, \FB_{\ex} $ in \cite{BBS}. In particular, the latter construction renders some additional vector spaces $ \Ren_1, \Ren, \eRen $ as a byproduct. Let us quickly comment on how to find certain elements of these spaces in $ \EB $:\\

Elements $ \fr \in \Ren_1 $ are a rigorous implementation of formal sums
\begin{equation}
	\Psi_\fr := \sum_{j \in \NNN} \phi_j
\label{eq:fr} 
\end{equation}
for some fixed $ \bphi \in \cE $. Mathematically, $ \fr = [ \bphi ] \subset \cE $ is an equivalence class with
\begin{equation}
	\bphi' \in [ \bphi ] \quad \Leftrightarrow \quad \sum_j (\phi_j - \phi'_j) = 0,
\label{eq:phiphiequivalence}
\end{equation}
where the sum is required to be absolutely convergent. It is easy to see that \eqref{eq:fr} is a sum notation of the coset\footnote{Here, $ [\Psi] $ denotes the coset of $ \cI $ in $ \EB $, which contains $ \Psi \in \EBn $, and $ \delta $ is the Kronecker delta.} $ \Psi_\fr = [(N, L) \mapsto \delta_{N, 0} \delta_{L, 1} \bphi] \in \EB $, so the pair $ (0, 1) $ is mapped to $ (\Psi_\fr)^{(1)}_{(1)}(j) = \phi_j $ and all other $ (N, L) $ are mapped to 0. Further, for $ \Psi'_\fr := \sum_{j \in \NNN} \phi'_j \in \EB $, we have $ \Psi'_\fr = \Psi_\fr $ if and only if \eqref{eq:phiphiequivalence} holds, since
\begin{equation}
	\sum_{j \in \NNN} \phi_j = \sum_{j \in \NNN} \phi'_j + \sum_{j \in \NNN} (\phi_j - \phi'_j),
\end{equation}
by definition of addition in $ \EBn $.\\

Elements $ \fR \in \Ren $ are now linear combinations of products of the form $ \fr_1 \ldots \fr_P $. These products rigorously implement formal sums
\begin{equation}
	\Psi_\fR := \sum_{j_1 \ldots j_P} (\bphi_1)_{j_1} \ldots (\bphi_P)_{j_P},
\end{equation} 
with $ \bphi_p \in \cE $ being a representative function of $ \fr_p $. It is easy to see that $ \Psi_\fR \in \EB $ in the sum notation.\\

Finally, elements $ e^\fr \in \eRen $ can be seen as a rigorous implementation of
\begin{equation}
	e^{\Psi_\fr} = \sum_{N = 0}^\infty \frac{\Psi_\fr^N}{N!} = \sum_{N = 0}^\infty \sum_{j_1, \ldots, j_N} \frac{1}{N!} \phi_{j_1} \cdot \ldots \cdot \phi_{j_N}.
\label{eq:efr}
\end{equation}
Obviously, $ e^{\Psi_\fr} \in \EB $. However, generic elements of $ \fc \in \eRen $ are quotients of the form $ \fc = \frac{a_1}{a_2} $ with $ a_j $ being a linear combination of terms $ e^{\fr} $ with $ \fr \in \Ren_1 $. Such a quotient can generally not be written as a sum of the form \eqref{eq:Psisum}. So we cannot readily find $ \eRen $ in $ \EB $ and thus also not the $ \eRen $--vector spaces $ \FB $ and $ \FB_{\ex} $.\\

\end{enumerate}	
\setcounter{remarks}{\theenumi}

\section{Main Result}	
\label{sec:mainresult}

Consider any Bogoliubov transformation $ \cV = \left( \begin{smallmatrix} u & v \\ \overline{v} & \overline{u} \end{smallmatrix} \right) $ with $ u, v: \cD \to \ell^2 $ (recall: $ \cD $ in \eqref{eq:cD} is the set of all sequences with finite support), as well as the dense domain
\begin{equation}
	\cD_\sF := \mathrm{span} \big\{a^\dagger(\bphi_1) \ldots a^\dagger(\bphi_N) \Omega, \; N \in \NNN_0, \; \bphi_{\ell} \in \cD \big\},
\label{eq:cDsF}
\end{equation}
where $ \Omega \in \sF $, defined by $\Omega^{(0)} = 1 $ and $ \Omega^{(N)} = 0 $ for $ N \ge 1 $, is the vacuum vector. In analogy to \cite[Def.~5.1]{BBS}, we define extended implementability as follows.

\begin{definition}
	A linear injective operator $ \UUU_\cV: \cD_\sF \to \EB $ \textbf{implements} a Bogoliubov transformation $ \cV $ \textbf{in the extended sense} if and only if
\begin{equation}
	\UUU_\cV a^\dagger(\bphi) \UUU_\cV^{-1} \Psi = b^\dagger(\bphi) \Psi, \qquad \UUU_\cV a(\bphi) \UUU_\cV^{-1} \Psi = b(\bphi) \Psi \qquad \forall \; \bphi \in \cD, \; \Psi \in \UUU_\cV(\cD_\sF).
\label{eq:implementation}
\end{equation}
In that case, $ \UUU_\cV $ is called an \textbf{extended implementer} of $ \cV $ and $ \cV $ is called \textbf{implementable in the extended sense}.
\label{def:implementation}
\end{definition}

Our main result is the following.

\begin{theorem}[Bosonic Extended Implementation Works]
Any bosonic Bogoliubov transformation $ \cV = \left( \begin{smallmatrix} u & v \\ \overline{v} & \overline{u} \end{smallmatrix} \right) $ with $ v, u: \cD \to \ell^2 $ is implementable in the extended sense.
\label{thm:bosonicESS}
\end{theorem}

\paragraph{Remarks.} 

\begin{enumerate}
\setcounter{enumi}{\theremarks}
\item \label{rem:regular_implementability} \textit{Regular implementability}: Definition \ref{def:implementation} generalizes the (regular) notion of implementability of $ \cV $ on Fock space. In the regular sense, $ \cV $ is implementable if and only if there exists some unitary $ \UUU_\cV: \sF \to \sF $ satisfying \eqref{eq:implementation}. By taking limits in $ \ell^2 $, the case $ \bphi \in \cD $ can then be generalized to $ \bphi \in \ell^2 $.\\
While regular implementability only holds for $ \tr(v^* v) = \tr(v^T \overline{v}) < \infty $ \cite{Shale1962, ShaleStinespring1965}, extended implementability can be achieved for any $ v^* v $, as long as $ v $ is defined on $ \cD $.\\

\item \label{rem:diagonalization} \textit{Diagonalization}: Extended implementers can be used to diagonalize quadratic Hamiltonians in some extended sense, as described in \cite[Sect.~6]{BBS}: Consider a diagonalizable quadratic Hamiltonian $ H \in \cAB $, that is,
\begin{equation}
	H = \frac{1}{2} \sum_{j,k \in \NNN} (2 h_{jk} a^\dagger_j a_k \mp k_{jk} a^\dagger_j a^\dagger_k + \overline{k_{jk}} a_j a_k).
\label{eq:Hquadratic}
\end{equation}
In certain cases \cite{Araki1968, NamNapiorkowskiSolovej2016}, there is an algebraic Bogoliubov transformation $ \cV_\cA $ and a normal ordering constant $ c $, such that $ (H + c) $ becomes diagonal in the $ b, b^\dagger $--operators, so $ \widetilde{H} = \cV_\cA^{-1} (H + c) $ is block--diagonal in the $ a, a^\dagger $--operators:
\begin{equation}
	(H + c) = \sum_{j,k \in \NNN} E_{jk} b^\dagger_j b_k \qquad \Rightarrow \qquad
	\widetilde{H} = \sum_{j,k \in \NNN} E_{jk} a^\dagger_j a_k = d \Gamma(E)
\end{equation}
and such that $ E $ is self--adjoint on some domain $ \dom(E) \subseteq \ell^2 $. In that case, $ d \Gamma(E) $ is also self--adjoint. Further, if $ E $ maps $ \cD \to \cD $, then
\begin{equation}
	\widetilde{H} = \UUU_\cV^{-1} (H + c) \UUU_\cV
\end{equation}
holds as a strong operator identity $ \cD_\sF \to \cD_\sF $. So $ H $ is diagonalized in the extended sense by $ \UUU_\cV $. Here $ c $ is a possibly divergent sum of the form \eqref{eq:fr}, so we can interpret $ c \in \EB $, and a multiplication by $ c $ maps $ \EB \to \EB $.\\

\end{enumerate}	
\setcounter{remarks}{\theenumi}

\section{Proof of the Main Result}
\label{sec:proofmainresult}

\subsection{General Construction of the Extended Implementer}	
\label{subsec:extimplementer}

We now construct the extended implementer $ \UUU_\cV: \cD_\sF \to \EB $ for a given Bogoliubov transformation $ \cV $. The construction employs creation and annihilation operators $ a^\sharp(\bphi) $, which we first need to extend to $ \EB $ in such a way that they satisfy the CCR. Consider an element $ [\Psi] \in \EB $ characterized by a coset element (``representative'') $ \Psi \in \EBn, \Psi: (N, L) \mapsto \Psi_{(L)}^{(N + L)} \in \cE^{(N + L)} $. In analogy to \eqref{eq:aadaggermomentumspace}, we then set
\begin{equation}
\begin{aligned}
	(a^\dagger(\bphi) \Psi)_{(L)}^{(N + L)}(q, q') &:= \sum_{k = 1}^N \frac{1}{\sqrt{N}} \phi_{j_k} \Psi_{(L)}^{(N + L - 1)}(q \setminus j_k, q')\\
	(a(\bphi) \Psi)_{(L + 1)}^{(N + L + 1)}(q, j, q'') &:= \sqrt{N+1} \; \overline{\phi_{j}} \Psi_{(L)}^{(N + 1 + L)}(q, j, q''), \qquad (a(\bphi) \Psi)_{(0)}^{(N)} := 0,
\end{aligned}
\label{eq:aadaggerESS}
\end{equation}
with $ q = (j_1, \ldots, j_N), q' = (j_{N + 1}, \ldots, j_{N + L}), j = j_{N + 1} $ and $ q'' = (j_{N + 2}, \ldots, j_{N + L + 1}) $.\\

\begin{lemma}[Operator Extensions]
For $ \bphi \in \cE $, the representative--wise definition \eqref{eq:aadaggerESS} renders well--defined linear operators
\begin{equation}
	a^\dagger(\bphi): \EB \to \EB, \qquad
	a(\bphi): \EB \to \EB.
\end{equation}
\label{lem:aadaggerwelldefined}
\end{lemma}
\begin{proof}
Consider the two expressions in \eqref{eq:aadaggerESS}. It is clear that the functions $ (N, L) \mapsto (a^\sharp(\bphi) \Psi)_{(L)}^{(N + L)} $ are elements of $ \EBn $. So it remains to show that for $ \Psi_I \in \cI $, we also have $ a^\sharp(\bphi) \Psi_I \in \cI $.\\
First, suppose that $ \Psi_I = \Psi - \Psi' $ where $ (\Psi')_{(L)}^{(N + L - 1)}(q \setminus j_k, q') $ is obtained from $ \Psi_{(L)}^{(N + L - 1)}(q \setminus j_k, q') $ by a permutation of the last $ L $ indices (i.e., those in $ q' $). This situation corresponds to case (A) above \eqref{eq:EB}. Then the same index permutation in $ q' $ transforms $ (a^\dagger(\bphi) \Psi)_{(L)}^{(N + L)}(q, q') $ into $ (a^\dagger(\bphi) \Psi')_{(L)}^{(N + L)}(q, q') $. Thus, $ a^\dagger(\bphi) \Psi_I \in \cI $. A similar argument shows that $ a(\bphi) \Psi_I \in \cI $ when $ (\Psi')_{(L)}^{(N + 1 + L)}(q, j, q'') $ is obtained from $ \Psi_{(L)}^{(N + 1 + L)}(q, j, q'') $ by permuting the last $ L $ indices (i.e., those in $ q'' $).\\
Concerning case (B), suppose there was a choice of $ 0 \le \Delta L_{N, L} \le L $ for each $ N, L \in \NNN_0 $, such that \eqref{eq:PsiPsiprime2} was true. Then,
\begin{equation}
\begin{aligned}
	&(a^\dagger(\bphi) \Psi')_{(L')}^{(N + L')}(q, j_{N + 1}, \ldots, j_{N + L'})\\
	\overset{\eqref{eq:aadaggerESS}}{=} &\sum_{k = 1}^N \frac{1}{\sqrt{N}} \phi_{j_k} \; (\Psi')_{(L')}^{(N + L' - 1)}(q \setminus j_k, j_{N + 1}, \ldots, j_{N + L'})\\
	\overset{\eqref{eq:PsiPsiprime2}}{=} &\sum_{k = 1}^N \frac{1}{\sqrt{N}} \phi_{j_k} \sum_{L: L' + \Delta L_{N, L} = L} \; \sum_{j_{N + L' + 1}, \ldots, j_{N + L}} \Psi_{(L)}^{(N + L - 1)}(q \setminus j_k, j_{N + 1}, \ldots, j_{N + L})\\
	\overset{\eqref{eq:aadaggerESS}}{=} &\sum_{L: L' + \Delta L_{N, L} = L} \; \sum_{j_{N + L' + 1}, \ldots, j_{N + L}} (a^\dagger(\bphi) \Psi)_{(L)}^{(N + L)}(j_1, \ldots, j_{N + L}) \in \CCC,
\end{aligned}
\end{equation}
so $ a^\dagger(\bphi) \Psi_I \in \cI $. A similar calculation with integrating out up to $ L $ indices of $ q'' $ shows that $ a(\bphi) \Psi_I \in \cI $.\\
Now, since $ \cI $ is generated by elements of type (A) and (B), we conclude that $ \Psi_I \in \cI $ always implies $ a^\dagger(\bphi) \Psi_I \in \cI $ and $ a(\bphi) \Psi_I \in \cI $, so the definitions of $ a^\dagger(\bphi) $ and $ a(\bphi) $ do not depend on the choice of the coset representative $ \Psi $.\\
\end{proof}

\begin{lemma}[Extended CCR]
	$ a^\dagger(\bphi), a(\bphi) $, as defined in Lemma \ref{lem:aadaggerwelldefined}, satisfy the extended CCR:
\begin{equation}
	[a(\bphi), a(\bpsi)] = [a^\dagger(\bphi), a^\dagger(\bpsi)] = 0, \qquad
	[a(\bphi), a^\dagger(\bpsi)] = \langle \bphi, \bpsi \rangle = \sum_j \overline{\phi_j} \psi_j \in \EB
\label{eq:extendedCCR}
\end{equation}
as a strong operator identity on $ \EB $ for $ \bphi, \bpsi \in \cE $.
\label{lem:extendedCCR}
\end{lemma}
\begin{proof}
This follows by a direct calculation using basis coefficients, as in the case $ \bphi, \bpsi \in \ell^2 $.\\
\end{proof}

\begin{lemma}[Bogoliubov Transformations Conserve Extended CCR]
Consider a Bogoliubov transformation $ \cV $ satisfying the Bogoliubov relations \eqref{eq:Bogoliuborelations}. Then the Bogoliubov--transformed operators $ b^\sharp $ in \eqref{eq:bdaggerb} still satisfy the extended CCR
\begin{equation}
	[b(\bphi), b(\bpsi)] = [b^\dagger(\bphi), b^\dagger(\bpsi)] = 0, \qquad
	[b(\bphi), b^\dagger(\bpsi)] = \langle \bphi, \bpsi \rangle
\label{eq:extendedCCRb}
\end{equation}
as a strong operator identity on $ \EB $ for $ \bphi, \bpsi \in \cD $.\\
\label{lem:extendedCCRb}
\end{lemma}
\begin{proof}
	Lemma \ref{lem:extendedCCR} ensures that the CCR hold for $ a^\sharp(\bphi): \EB \to \EB $. Lemma \ref{lem:aadaggerwelldefined} renders well--definedness of $ b^\sharp(\bphi): \EB \to \EB $ for $ \bphi \in \cD $.  The recovery of the CCR for $ b^\sharp $ using the Bogoliubov relations \eqref{eq:Bogoliuborelations} works as in the case of $ a^\sharp, b^\sharp $ being Fock space operators. For instance, we have\footnote{It is easy to verify that the first two expressions are well--defined operators $ \EB \to \EB $. The other expressions are elements of $ \EB $, which can be interpreted as multiplication operators $ \EB \to \EB $, since the $ \otimes $--product of two $ \EB $--elements is again in $ \EB $, see Remark \ref{rem:tensorproduct}.}
\begin{equation}
\begin{aligned}
	&{[b(\bphi), b^\dagger(\bpsi)]}
	= [a(u \bphi) + a^\dagger(v \overline{\bphi}), a^\dagger(u \bpsi) + a(v \overline{\bpsi})]
	\overset{\text{Lemma \ref{lem:extendedCCR}}}{=} \langle u \bphi, u \bpsi \rangle - \langle v J \bpsi, v J \bphi \rangle\\
	=~&\langle \bphi, u^* u \bpsi \rangle - \langle J J v^* v J \bpsi, J \bphi \rangle
	= \langle \bphi, u^* u \bpsi \rangle - \langle \bphi, v^T \overline{v} \bpsi \rangle
	\overset{\eqref{eq:Bogoliuborelations}}{=} \langle \bphi, \bpsi \rangle.
\end{aligned}
\end{equation}
The other two identities in \eqref{eq:extendedCCRb} are obtained analogously.\\
\end{proof}

As in the case $ \tr(v^* v) < \infty $ ~\cite{Solovej2014}, the definition of our extended implementer is based on a vector $ \Omega_\cV \in \EB $, called Bogoliubov vacuum. We first give a definition of $ \UUU_\cV $ for a given $ \Omega_\cV $ and then construct $ \Omega_\cV $ further below. Recall the definition of $ \cD_\sF $ \eqref{eq:cDsF}.\\

\begin{definition}
Given a Bogoliubov--transformed vacuum state $ \Omega_\cV \in \EB $, we define the linear \textbf{extended Bogoliubov implementer} $ \UUU_\cV: \cD_\sF \to \EB $ by
\begin{equation}
	\UUU_\cV a^\dagger(\bphi_1) \ldots a^\dagger(\bphi_n) \Omega := b^\dagger(\bphi_1) \ldots b^\dagger(\bphi_n) \Omega_\cV,
\label{eq:transformbogoliubovstate}
\end{equation}
with $ \bphi_{\ell} \in \cD $ and $ b^\dagger(\bphi_j) = (a^\dagger(u \bphi_j) + a(v \overline{\bphi_j})) $.\\
\label{def:implementer}
\end{definition}
It follows from Lemma \ref{lem:aadaggerwelldefined} that the right--hand side of \eqref{eq:transformbogoliubovstate} is an element of $ \EB $, so $ \UUU_\cV $ is well--defined.\\

\subsection{Construction of the Bogoliubov vacuum}
\label{subsec:Bogoliubovvacuum}

The only remaining step for finishing the construction of $ \UUU_\cV $ is to provide a reasonable Bogoliubov vacuum $ \Omega_\cV \in \EB $, which we do in Definition \ref{def:bosonicbogoliubovvacuum}. By ``reasonable'', we mean that $ \Omega_\cV $ is annihilated by all $ b $--operators, which we prove in Lemma \ref{lem:bOmegacVbosonic}. This property will play an important role in the proof of our main theorem.\\

Let us quickly explain the heuristics for the choice of $ \Omega_\cV $. In case $ \tr(v^* v) < \infty $, the Bogoliubov vacuum is well--known \cite[(61)]{Solovej2014}: Up to normalization, it consists of an exponential over two--particle wave functions of the kind
\begin{equation}
	K^{(2)} = - \sum_{j \in \NNN} \frac{\nu_j}{2 \mu_j} \bg_j \otimes \bg_j,
\label{eq:bosonicpair}
\end{equation}
where $ j $ indexes a simultaneous eigenbasis $ (\bof_j)_{j \in \NNN} $ of both\footnote{We have $ v^T \overline{v} u^* v J = u^* u u^* v J - u^* v J = u^* v v^* v J = u^* v J v^T \overline{v} $, so $ [v^T \overline{v}, u^* v J] = 0 $. Further, for $ \tr(v^* v) < \infty $, both operators are Hilbert--Schmidt, which allows for a simultaneous eigenbasis.} $ v^T \overline{v} $ and $ u^* v J $ with eigenvalues $ \lambda_j $ (with respect to $ v^T \overline{v} $), and where $ \nu_j = \sqrt{\lambda_j} $ and $ \mu_j = \sqrt{1 + \lambda_j} $. The orthonormal basis $ (\bg_j)_{j \in \NNN} $ is then given by $ \bg_j = \lambda_j^{-1/2} v J \bof_j $. Using that $ u^* v J $ is a spectral multiplication by $ \nu_j \mu_j = \sqrt{\lambda_j (1 + \lambda_j)} $, it is now an easy task to verify that
\begin{equation}
	\langle u \bphi, \bg_j \rangle = \langle \bphi, u^* v J \lambda_j^{-1/2} \bof_j \rangle = \mu_j \langle \bphi, \bof_j \rangle \quad \Rightarrow \quad
	2 \frac{a(u \bphi)}{\sqrt{2}} K^{(2)} = - v J \bphi \qquad \forall \bphi \in \cD.
\end{equation}
So when comparing with \eqref{eq:aadaggermomentumspace}, we see that $ K^{(2)} $ is the integral kernel of an operator which turns $ 2 J u \bphi $ into $ - v J \bphi $.

When generalizing to the case with arbitrary spectrum $ \sigma := \sigma(v^* v) = \sigma(v^T \overline{v}) $, an eigenbasis $ (\bg_j)_{j \in \NNN} $ of $ v^T \overline{v} $ will generally no longer exist. However, we still have a projection--valued measure (PVM) $ P_{v^T \overline{v}} $ with
\begin{equation}
	v^T \overline{v} = \int_\sigma \lambda \; dP_{v^T \overline{v}}(\lambda).
\label{eq:dPvv}
\end{equation}
And we are still able to define an operator $ \cO $, which turns $ 2 J u \bphi $ into $ - v J \bphi $ and has an integral kernel $ K_\cO \in \cE^{(2)} $: Consider the polar decompositions with respect to $ P_{v^T \overline{v}} $,
\begin{equation}
	v J = A_v \sqrt{\lambda}, \qquad
	J u = A_u \sqrt{1 + \lambda},
\label{eq:polardecomp}
\end{equation}
following from $ v^T \overline{v} = J v^* v J = u^* u $, where functions of $ \lambda $ are to be understood as spectral multiplications and $ A_v, A_u : \ell^2 \to \ell^2 $ are anti--unitary operators. Then, $ \cO: \ell^2 \to \ell^2 $ is defined as
\begin{equation}
	\cO = - A_v \sqrt{\frac{\lambda}{4 (1 + \lambda)}} A_u^*,
\label{eq:cOpolar}
\end{equation}
which is clearly bounded by $ 1/2 $.

\begin{lemma}[Bosonic Pairs]
	For any $ \bphi \in \ell^2 $, we have
\begin{equation}
	2 \cO J u \bphi = - v J \bphi \qquad \text{and} \qquad \Vert \cO \bphi \Vert < \frac{1}{2} \Vert \bphi \Vert.
\label{eq:cO}
\end{equation}
Further, $ \cO $ allows for an integral kernel $ K_\cO \in \cE^{(2)} $, so
\begin{equation}
	(\cO \bphi)_j = \sum_{j'} K_\cO(j, j') \phi_{j'}.
\label{eq:KcO}
\end{equation}
\label{lem:bosonicpair}
\end{lemma}
\begin{proof}
The equality in \eqref{eq:cO} is verified by plugging in the polar decompositions \eqref{eq:polardecomp} and \eqref{eq:cOpolar}. $ \Vert \cO \bphi \Vert < \frac{1}{2} \Vert \bphi \Vert $ is a direct consequence of $ \sqrt{\frac{\lambda}{4 (1 + \lambda)}} < \frac{1}{2} $. Restricting $ \cO $ to $ \cD \subset \ell^2 $, we get the operator $ \cO|_{\cD} : \cD \to \ell^2 \subset \cE $. By the Schwartz kernel theorem, every operator $ \cD \to \cE $ has an integral kernel, which implies the existence of $ K_\cO $.
\end{proof}

\begin{definition}[Bosonic Bogoliubov Vacuum]
We define $ \Omega_\cV \in \EB $ by
\begin{equation}
	(\Omega_\cV)_{(0)}^{(2m)} := \frac{\sqrt{(2m)!}}{m!} (K_\cO)^{\otimes_S m}, \qquad (\Omega_\cV)_{(0)}^{(2m + 1)} = 0
\label{eq:bosonicbogoliubovvacuum}
\end{equation}
for $ m \in \NNN_0 $ and $ (\Omega_\cV)_{(L)}^{(N + L)} = 0 $ for $ L \ge 1 $.
\label{def:bosonicbogoliubovvacuum}
\end{definition}
\noindent In what follows, only the sector $ L = 0 $ will be relevant, so we will drop the index $ (L) $. Further, $ \Omega_\cV \in \cE_\sF $ under the embedding $ \cE_\sF \to \EB $ described in Remark \ref{rem:Fockspaceextension}.\\

\begin{lemma}[$ b $ Annihilates Bosonic $ \Omega_\cV $]
For $ \bphi \in \cD $, we have
\begin{equation}
	b(\bphi) \Omega_\cV = 0.
\label{eq:bannihilationbosonic}
\end{equation}
\label{lem:bOmegacVbosonic}
\end{lemma}
\begin{proof}
Recall \eqref{eq:bdaggerb} and the definition of the antilinear conjugation operator $ J $, which imply
\begin{equation}
	b(\bphi) \Omega_\cV = a(u \bphi) \Omega_\cV + a^\dagger(v J \bphi) \Omega_\cV.
\end{equation}
The creation term has only contributions of odd sector numbers $ (N) = (2m+1) $ that evaluate to
\begin{equation}
	(a^\dagger(v J \bphi) \Omega_\cV)^{(2m+1)} = \sqrt{2m+1} (v J \bphi) \otimes_S (\Omega_\cV^{(2m)}) = \frac{\sqrt{(2m+1)!}}{m!} (v J \bphi) \otimes_S (K_\cO)^{\otimes_S m}.
\label{eq:adaggervJ}
\end{equation}
The annihilation term also has only contributions of odd $ (N) $, which can be evaluated in mode--configuration space. In the following, let $ \sigma $ and $ \tilde\sigma $ be permutations of $ \{1, \ldots, 2m+2\} $ and $ \{1, \ldots, 2m+1\} $, let $ j = j_{\sigma(2m+2)} $, and denote by $ \sum_\sigma, \sum_{\tilde\sigma} $ the sums over all $ (2m+2)! $ or $ (2m+1)! $ permutations. Then,
\begin{equation}
\begin{aligned}
	&(a(u \bphi) \Omega_\cV)^{(2m+1)}(j_1, \ldots, j_{2m+1})\\
	= & \sqrt{2m+2} \sum_j (\overline{u \bphi})_j \Omega_\cV^{(2m+2)}(j_1, \ldots, j_{2m+1}, j)\\
	= & \frac{\sqrt{2m+2}}{(2m+2)!} \frac{\sqrt{(2m+2)!}}{(m+1)!} \sum_\sigma \sum_j (\overline{u \bphi})_j \prod_{k = 1}^{m+1} K_\cO(j_{\sigma(2k-1)}, j_{\sigma(2k)})\\
	= & \frac{\sqrt{(2m+1)!}}{(m+1)!} \frac{1}{(2m+1)!} \left( (m+1) \sum_{\tilde\sigma} \sum_j (\overline{u \bphi})_j K_\cO(j, j_{\tilde\sigma(2m+1)}) \prod_{k = 1}^m K_\cO(j_{\tilde\sigma(2k-1)}, j_{\tilde\sigma(2k)}) \right.\\
	&~~~~~~~~~~~~~~~~~~~~~~~~~~~~~ \left. + (m+1) \sum_{\tilde\sigma} \sum_j (\overline{u \bphi})_j K_\cO(j_{\tilde\sigma(2m+1)}, j) \prod_{k = 1}^m K_\cO(j_{\tilde\sigma(2k-1)}, j_{\tilde\sigma(2k)}) \right)\\
	&\Rightarrow \quad (a(u \bphi) \Omega_\cV)^{(2m+1)} = \frac{\sqrt{(2m+1)!}}{m!} \left( \sum_j (J u \bphi)_j (K_\cO(j, \cdot) +  K_\cO(\cdot, j) ) \right) \otimes_S K_\cO^{\otimes_S m}.
\end{aligned}
\label{eq:KcOsum}
\end{equation}
In order to evaluate the term in brackets, we use that $ K_\cO $ is the integral kernel of $ \cO $, so
\begin{equation}
	\sum_j (J u \bphi)_j K_\cO(\cdot, j) = \cO J u \bphi
	\overset{\eqref{eq:cO}}{=} - \frac{1}{2} v J \bphi.
\label{eq:KcO1term}
\end{equation}
For evaluating the other term, we use that the integral kernel of $ \cO^*: \ell^2 \to \ell^2 $ is $ K_{\cO^*}(j, j') = \overline{K_\cO(j', j)} $, so
\begin{equation}
	\sum_j (J u \bphi)_j K_\cO(j, \cdot) = \sum_j \overline{ K_{\cO^*}(\cdot, j)} (J u \bphi)_j = J \cO^* J J u \bphi = J \cO^* u \bphi.
\end{equation}
The Bogoliubov relations \eqref{eq:Bogoliuborelations} and polar decompositions \eqref{eq:polardecomp} now imply
\begin{equation}
\begin{aligned}
	(u^* v)^* = J u^* v J \quad &\Leftrightarrow \quad
	J^* \sqrt{\lambda} A_v^* J A_u \sqrt{1 + \lambda} = J \sqrt{1 + \lambda} A_u^* J^* A_v \sqrt{\lambda}\\
	 \quad &\Leftrightarrow \quad J^* A_u \sqrt{\frac{1}{1 + \lambda}} J J^* \sqrt{\lambda} A_v^* = A_v \sqrt{\frac{\lambda}{1 + \lambda}} A_u^* J = - \cO J,
\end{aligned}
\label{eq:cOJ}
\end{equation}
with $ \sqrt{\frac{1}{1 + \lambda}} $ being a densely defined spectral multiplication operator. Now, if $ B $ is some densely defined operator on $ \ell^2 $ and $ B^* $ is its adjoint, then for $ J \bpsi \in \dom(B) $ and $ J^* \bphi \in \dom(B^*) $, we have
\begin{equation}
\begin{aligned}
	\langle \bphi, J B J \bpsi \rangle
	=~ &\langle J J B^* J^* \bphi, J \bpsi \rangle
	= \langle \bpsi, J B^* J^* \bphi \rangle
	= \langle J B J^* \bpsi, J J \bphi \rangle
	= \langle \bphi, J^* B J^* \bpsi \rangle\\
	\quad \Rightarrow \quad J^* B J^* =~ &J B J.
\label{eq:JBJ}
\end{aligned}
\end{equation}
So 
\begin{equation}
\begin{aligned}
	&J^* A_u \sqrt{\frac{1}{1 + \lambda}} J J^* \sqrt{\lambda} A_v^*
	= J A_u \sqrt{\frac{1}{1 + \lambda}} J J \sqrt{\lambda} A_v^*
	= J A_u \sqrt{\frac{\lambda}{1 + \lambda}} A_v^* = - J \cO^*\\
	&\overset{\eqref{eq:cOJ}}{\Rightarrow} \quad J \cO^* = \cO J
	\quad \Rightarrow \quad
	\sum_j (J u \bphi)_j K_\cO(j, \cdot) = - \frac{1}{2} v J \bphi.
\end{aligned}
\label{eq:KcO2term}
\end{equation}
Plugging \eqref{eq:KcO1term} and \eqref{eq:KcO2term} into \eqref{eq:KcOsum}, we obtain
\begin{equation}
	(a(u \bphi) \Omega_\cV)^{(2m+1)} = - \frac{\sqrt{(2m+1)!}}{m!} (v J \bphi) \otimes_S K_\cO^{\otimes_S m},
\label{eq:KcOannihilation}
\end{equation}
which exactly cancels \eqref{eq:adaggervJ} and establishes the lemma.
\end{proof}



\subsection{Conditions for Extended Implementability}
\label{subsec:conditionsextimpl}

Before proceeding to the final proof of Theorem \ref{thm:bosonicESS}, we first set up some simple conditions for when $ \UUU_\cV $ is an extended implementer. These conditions are analogous to the ones given in \cite[Lemma~5.2]{BBS}.\\

\begin{lemma}[Conditions for an Extended Implementer]
	Let $ \Omega_\cV \in \EB $ such that $ b_j \Omega_\cV = 0 \; \forall j \in \NNN $. Further, let $ \UUU_\cV $ as in Definition \ref{def:implementer} be injective (so $ \UUU_\cV^{-1} $ exists). Then, $ \UUU_\cV $ implements $ \cV $ in the extended sense.
\label{lem:implementationconditions}
\end{lemma}
\begin{proof}
In order to establish the extended implementation \eqref{eq:implementation}, it suffices to show
\begin{equation}
	\UUU_\cV a^\sharp_j \UUU_\cV^{-1} \Psi = b^\sharp_j \Psi,
\end{equation}
for $ a^\sharp \in \{a^\dagger, a\} $ and
\begin{equation}
	\Psi = \UUU_\cV a^\dagger_{j_1} \ldots a^\dagger_{j_N} \Omega
	\overset{\eqref{eq:transformbogoliubovstate}}{=} b_{j_1}^\dagger \ldots b_{j_N}^\dagger \Omega_\cV,
\end{equation}
where $ N \in \NNN_0 $ and $ j, j_1, \ldots, j_N \in \NNN $ are arbitrary. In case $ a^\sharp_j = a^\dagger_j $, we directly compute
\begin{equation}
	\UUU_\cV a^\dagger_j \UUU_\cV^{-1} \Psi
	= \UUU_\cV a^\dagger_j a^\dagger_{j_1} \ldots a^\dagger_{j_N} \Omega
	\overset{\eqref{eq:transformbogoliubovstate}}{=} b^\dagger_j b^\dagger_{j_1} \ldots b^\dagger_{j_N} \Omega_\cV
	= b^\dagger_j \Psi.
\end{equation}
In case $ a^\sharp_j = a_j $, we obtain
\begin{equation}
\begin{aligned}
	\UUU_\cV a_j \UUU_\cV^{-1} \Psi	= &\UUU_\cV a_j a^\dagger_{j_1} \ldots a^\dagger_{j_N} \Omega\\
	= &\UUU_\cV \sum_{k = 1}^N a^\dagger_{j_1} \ldots a^\dagger_{j_{k-1}} [a_j, a^\dagger_{j_k}] a^\dagger_{j_{k+1}} \ldots a^\dagger_{j_N} \Omega + \UUU_\cV  a^\dagger_{j_1} \ldots a^\dagger_{j_N} \underbrace{a_j \Omega}_{=0} \\
	= &\sum_{k = 1}^N \delta_{j j_k} \UUU_\cV a^\dagger_{j_1} \ldots a^\dagger_{j_{k-1}} a^\dagger_{j_{k+1}} \ldots a^\dagger_{j_N} \Omega\\
	\overset{\eqref{eq:transformbogoliubovstate}}{=} &\sum_{k = 1}^N \delta_{j j_k} b^\dagger_{j_1} \ldots b^\dagger_{j_{k-1}} b^\dagger_{j_{k+1}} \ldots b^\dagger_{j_N} \Omega_\cV.
\end{aligned}
\label{eq:transformationCCR1}
\end{equation}
On the other hand, conservation of the CCR implies that $ [b_j, b_{j'}^\dagger] = \delta_{jj'} $, so
\begin{equation}
\begin{aligned}
	b_j \Psi	
	\overset{\eqref{eq:transformbogoliubovstate}}{=} &b_j b^\dagger_{j_1} \ldots b^\dagger_{j_N} \Omega_\cV\\
	= &\sum_{k = 1}^N b^\dagger_{j_1} \ldots b^\dagger_{j_{k-1}} [b_j, b^\dagger_{j_k}] b^\dagger_{j_{k+1}} \ldots b^\dagger_{j_N} \Omega_\cV +  b^\dagger_{j_1} \ldots b^\dagger_{j_N} \underbrace{b_j \Omega_\cV}_{=0}\\
	= &\sum_{k = 1}^N \delta_{j j_k} b^\dagger_{j_1} \ldots b^\dagger_{j_{k-1}} b^\dagger_{j_{k+1}} \ldots b^\dagger_{j_N} \Omega_\cV.
\end{aligned}
\label{eq:transformationCCR2}
\end{equation}
Expressions \eqref{eq:transformationCCR1} and \eqref{eq:transformationCCR2} agree, which renders the desired equality.\\
\end{proof}

\subsection{Proof of Theorem \ref{thm:bosonicESS}}

\begin{proof}[Proof of Theorem \ref{thm:bosonicESS}]
We have to show that $ \cV $ has an implementer $ \UUU_\cV $, i.e., \eqref{eq:implementation} holds. The operator $ \UUU_\cV $ and the vector $ \Omega_\cV $ are given in Definitions \ref{def:implementer} and \ref{def:bosonicbogoliubovvacuum}, respectively. By Lemma \ref{lem:implementationconditions}, $ \UUU_\cV $ implements $ \cV $ if $ b_j \Omega_\cV = 0 \; \forall j \in \NNN $ and $ \UUU_\cV $ is injective. The property $ b_j \Omega_\cV = 0 $ readily follows from Lemma \ref{lem:bOmegacVbosonic}. So it remains to establish injectivity of $ \UUU_\cV: \cD_\sF \to \EB $ in order to finish the proof.\\

Since $ \cD_\sF $ is spanned by vectors of the type $ a_{j_1}^\dagger \ldots a_{j_N}^\dagger \Omega $, it suffices to show that the set
\begin{equation}
	\cB := \big\{ b_{j_1}^\dagger \ldots b_{j_N}^\dagger \Omega_\cV \; \big\vert \; N \in \NNN_0, j_k \in \NNN \big\} \subset \EB
\label{eq:bdaggerset}
\end{equation}
is linearly independent. To do so, we investigate the set $ \cB $ step by step.\\
Let us start by evaluating
\begin{equation}
	b_j^\dagger \Omega_\cV = (a^\dagger(u \be_j) + a(v J \be_j)) \Omega_\cV.
\end{equation}
For even sectors, we have $ (b_j^\dagger \Omega_\cV)^{(2m)} = 0 $. For evaluating the odd sectors, we use that for $ \bphi \in \ell^2 $, the same arguments as in \eqref{eq:KcOsum} through \eqref{eq:KcOannihilation} yield $ (a(\bphi) \Omega_\cV)^{(2m+1)} = \frac{\sqrt{(2m+1)!}}{m!} (2 \cO J \bphi) \otimes_S K_\cO^{\otimes_S m} $. Thus, together with the definition of $ \Omega_\cV $ \eqref{eq:bosonicbogoliubovvacuum}, we get
\begin{equation}
	(b_j^\dagger \Omega_\cV)^{(2m+1)}
	= \frac{\sqrt{(2m+1)!}}{m!} \underbrace{(u \be_j + 2 \cO J v J \be_j)}_{=: \bpsi_j} \otimes_S K_\cO^{\otimes_S m}
	= (a^\dagger(\bpsi_j) \Omega_\cV )^{(2m+1)}.
\label{eq:bdaggerOmegacV}
\end{equation}
Now, $ v J = -2 \cO J u $ implies
\begin{equation}
	\bpsi_j = u \be_j + 2 \cO J v J \be_j = (1 - 4 \cO J \cO J) u \be_j.
\label{eq:bpsi}
\end{equation}
By Lemma \ref{lem:bosonicpair}, we have $ \Vert \cO \bphi \Vert < \frac{1}{2} \Vert \bphi \Vert $, so $ \Vert 4 \cO J \cO J \bphi \Vert < \Vert \bphi \Vert $. Thus, the operator $ (1 - 4 \cO J \cO J) $ is injective, as is $ u $ (since $ u^* u \ge 1 $), and linear independence of $ (\be_j)_{j \in \NNN} $ implies linear independence of $ (\bpsi_j)_{j \in \NNN} $.\\
Now, let us turn to the evaluation of a general $ b_{j_1}^\dagger \ldots b_{j_N}^\dagger \Omega_\cV $. For $ N = 2 $,
\begin{equation}
\begin{aligned}
	b_{j_1}^\dagger b_{j_2}^\dagger \Omega_\cV
	\overset{\eqref{eq:bdaggerOmegacV}}{=} b_{j_1}^\dagger a^\dagger(\bpsi_{j_2}) \Omega_\cV
	=~ &a^\dagger(\bpsi_{j_2}) b_{j_1}^\dagger \Omega_\cV + [a^\dagger(u \be_{j_1}) + a(v J \be_{j_1}), a^\dagger(\bpsi_{j_2})] \Omega_\cV\\
	=~ &a^\dagger(\bpsi_{j_1}) a^\dagger(\bpsi_{j_2}) \Omega_\cV + \langle v J \be_{j_1}, \bpsi_{j_2} \rangle \Omega_\cV.
\end{aligned}
\end{equation}
By a similar expansion in terms of commutators, one easily sees that
\begin{equation}
	b_{j_1}^\dagger \ldots b_{j_N}^\dagger \Omega_\cV
	= a^\dagger(\bpsi_{j_1}) \ldots a^\dagger(\bpsi_{j_N}) \Omega_\cV + \mathrm{l.o.t.},
\end{equation}
where $ \mathrm{l.o.t} $ (``lower--order terms'') is a sum of expressions $ a^\dagger(\bpsi_{j_1'}) \ldots a^\dagger(\bpsi_{j'_{N'}}) \Omega_\cV $ with $ N < N' $ and $ \{j'_1, \ldots, j'_{N'} \} \subset \{j_1, \ldots, j_N\} $.\\

Now, linear independence of $ \cB $ will follow if we can prove linear independence of the set of ``leading--order terms''
\begin{equation}
	\sA := \big\{ a^\dagger(\bpsi_{j_1}) \ldots a^\dagger(\bpsi_{j_N}) \Omega_\cV \; \big\vert \; N \in \NNN_0, j_k \in \NNN \big\} \subset \cE_\sF.
\label{eq:adaggerset}
\end{equation}
To see this implication of linear independence, suppose, $ \cB $ would be linearly dependent, so there was a linear combination
\begin{equation}
	B = \sum_{k = 1}^K \lambda_k b_{j_{k,1}}^\dagger \ldots b_{j_{k,N_k}}^\dagger \Omega_\cV = 0,
\end{equation}
with $ \lambda_k \neq 0 $. By $ \overline{N} = \max_k N_k $ we denote the highest number of consecutively applied creation operators. Then, $ B $ amounts to a finite linear combination of elements of the kind $ a^\dagger(\bpsi_{j_1}) \ldots a^\dagger(\bpsi_{j_N}) \Omega_\cV \in \sA $, where the contribution of terms with $ N = \overline{N} $ is
\begin{equation}
	\sum_{k: N_k = \overline{N}} \lambda_k a^\dagger(\bpsi_{j_{k,1}}) \ldots a^\dagger(\bpsi_{j_{k,N_k}}) \Omega_\cV,
\end{equation}
with the sum being nonempty. As $ B = 0 $, linear independence of $ \sA $ would now imply that $ \lambda_k = 0 $ whenever $ N_k = \overline{N} $, which contradicts our premise $ \lambda_k \neq 0 $. So linear independence of $ \sA $ implies linear independence of $ \cB $.\\

Finally, we establish linear independence of $ \sA $ by a contradiction. Suppose there was a linear combination
\begin{equation}
	0 = \sum_{k = 1}^K \lambda_k a^\dagger(\bpsi_{j_{k,1}}) \ldots a^\dagger(\bpsi_{j_{k,N_k}}) \Omega_\cV,
\end{equation}
with $ \lambda_k \neq 0 $ and $ \underline{N} = \min_k N_k $ being the least number of consecutively applied creation operators. Then,
\begin{equation}
	(a^\dagger(\bpsi_{j_1}) \ldots a^\dagger(\bpsi_{j_N}) \Omega_\cV)^{(n)} = 0 \qquad \text{for } n < N,
\end{equation}
i.e., the $ N $ lowest sectors are unoccupied. Thus, the $ (\underline{N}) $--sector of our linear combination amounts to
\begin{equation}
	0 = \sum_{k: N_k = \underline{N}} \lambda_k (a^\dagger(\bpsi_{j_{k,1}}) \ldots a^\dagger(\bpsi_{j_{k,\underline{N}}}) \Omega_\cV)^{(\underline{N})}
	= \sqrt{\underline{N}!} \sum_{k: N_k = \underline{N}} \lambda_k \bpsi_{j_{k,1}} \otimes_S \ldots \otimes_S \bpsi_{j_{k,\underline{N}}}.
\end{equation}
Since $ \{ \bpsi_j \; \mid \; j \in \NNN \} $ is linearly independent in $ \ell^2 $, also the set
\begin{equation}
	\big\{ \bpsi_{j_{k,1}} \otimes_S \ldots \otimes_S \bpsi_{j_{k,\underline{N}}} \; \big\vert \; j_{k, \ell} \in \NNN \big\} \subset (\ell^2)^{\otimes \underline{N}}
\end{equation}
is linearly independent, which implies $ \lambda_k = 0 $ for all $ k $ with $ N_k = \underline{N} $, establishes the desired contradiction and finishes the proof.\\
\end{proof}

\bigskip


\noindent\textit{Acknowledgments.}
This paper originated from discussions with Andreas Deuchert at the INdAM Quantum Meetings 2022, which were supported by the Istituto Nazionale di Alta Matematica "F. Severi". The author was further financially supported by the Basque Government through the BERC 2018-2021 program, by the Ministry of Science, Innovation and Universities: BCAM Severo Ochoa accreditation SEV-2017-0718, as well as by the European Research Council (ERC) through the Starting Grant \textsc{FermiMath}, Grant
Agreement No. 101040991.\\

\end{document}